\documentclass{article}

\usepackage[a4paper, total={6in, 9in}]{geometry}

\usepackage{hyperref}
\usepackage{amssymb}
\usepackage{amsmath}
\usepackage{authblk}
\usepackage{amsthm}
\usepackage{physics}
\usepackage{todonotes}
\usepackage{comment}
\usepackage{tikz}
\usepackage{subcaption}
\usepackage{soul}

\usetikzlibrary{arrows.meta}

\usepackage{csquotes}

\usepackage[
  backend=biber,
  style=phys,
  sorting=none,
  giveninits=true,
  articletitle=true,
  biblabel=brackets,
  pageranges=false,
  chaptertitle=false,
  doi=false,
  url=false,
  isbn=false,
  hyperref=true,
  natbib=true
]{biblatex}

\newcommand{\F}{\mathcal{F}}
\newcommand{\W}{\mathcal{W}}
\newcommand{\intinf}{\int_{-\infty}^\infty}
\newcommand{\intoinf}{\int_0^\infty}
\newcommand{\ii}{\mathrm{i}}
\newcommand{\ee}{\mathrm{e}}
\newcommand{\xx}{\mathsf{x}}

\DeclareMathOperator{\sgn}{\text{sgn}}
\DeclareMathOperator{\sinc}{\text{sinc}}

\DeclareMathOperator{\supp}{\text{supp}}

\newtheorem{theorem}{Theorem}[section]
\newtheorem{proposition}[theorem]{Proposition}
\newtheorem{definition}[theorem]{Definition}

\definecolor{curve}{rgb}{0.541,0.012,0.071}

\numberwithin{equation}{section}
\title{When one sign is not enough:\\ $2+1$ circular motion Unruh effect at low energies}

\author[1]{Leo J. A. Parry\thanks{leo.parry@nottingham.ac.uk}}
\author[2,3]{Christopher J. Fewster\thanks{chris.fewster@york.ac.uk}}
\author[1]{Jorma Louko\thanks{jorma.louko@nottingham.ac.uk}}

\affil[1]{School of Mathematical Sciences,
University of Nottingham, 
\break Nottingham NG7 2RD, UK}
\affil[2]{Department of Mathematics, Ian Wand Building, Deramore Lane, \newline University of York, York YO10 5GH, UK}
\affil[3]{York Centre for Quantum Technologies, University of York, York YO10 5DD, UK}

\date{July 2026. Revised October 2026.\\[2ex]
{\small Published in Gen.\ Relativ. Gravit.\ \textbf{58}, 132 (2026). Version of record at \url{https://doi.org/10.1007/s10714-026-03629-9} with licence \url{http://creativecommons.org/licenses/by/4.0/}.}}

\begin{document}
\maketitle

\begin{abstract}
We address the circular motion Unruh effect in $2+1$ spacetime dimensions, as probed by a pointlike Unruh-DeWitt detector coupled to a massless scalar field. The effective temperature due to circular acceleration, operationally defined in terms of the detector's excitation and de-excitation probabilities, is known to be much smaller than the linear acceleration Unruh temperature when the detector's energy gap is small and the interaction lasts for a long time.
It was shown by Parry \emph{et al.\ }[Class.\ Quant.\ Grav.\ \textbf{42}, 245012 (2025)] that a temperature of the order of the linear acceleration Unruh temperature can nevertheless be recovered in a simultaneous long-time-small-gap double limit, using suitable classes of detector-field couplings described by asymptotically scaled switching families (ASSFs). The successful constructions presented there required the coupling to change sign.  
Here we prove, within the ASSF framework and under certain technical boundedness and localisation conditions, that sign changes in the detector-field coupling are in fact \emph{necessary\/} for obtaining a nonvanishing limiting effective temperature. 
Our analysis is motivated by current work towards an experimental verification of the circular motion Unruh effect in analogue spacetime experiments. 
\end{abstract}

\section{Introduction}

We are honoured to dedicate this paper to the life and science of Jerzy (Jurek) Lewandowski, who passed away prematurely in October 2024. Jurek's work combined mathematical precision with deep physical insight. 
In this paper, we address the circular motion Unruh effect for a massless scalar field in $2+1$ spacetime dimensions. We will begin by reviewing subtleties arising in the effective temperature at low energies, and the results
obtained in~\cite{Parry:2025wub}, which show how a nonvanishing effective temperature can be obtained in the low energy limit for certain 
asymptotically scaled switching families (ASSFs) that were introduced in~\cite{Parry:2025wub}. Here, the `switching' concerns the coupling between the detector and the quantum field along the detector's trajectory. In~\cite{Parry:2025wub}, we gave specific examples of sign-changing ASSFs that gave a nonvanishing effective temperature at low energies; here, we go further and establish 
a new theorem (Theorem~\ref{prop:necessary_sign_changes}), 
which shows that a sign change in the detector-field coupling is necessary to recover a nonvanishing effective temperature in the low energy limit, within the ASSF framework, 
under certain technical conditions. We believe that Jurek would have appreciated the precision that ASSFs afford to our new theorem. 

The context of and motivation for our work is the Unruh effect in quantum field theory. In its original form \cite{Fulling:1972md,Davies:1974th,Unruh1976,Fulling:2014}, one considers a uniformly linearly accelerated detector coupled to a relativistic quantum field in Minkowski space that is in its vacuum state. The Unruh effect predicts that the detector will thermalise at the Unruh temperature 
\begin{equation} 
\label{eq: Unruh temp}
    T_U = \frac{\hbar a}{2\pi c k_B} \,, 
\end{equation}
where $a$ is the proper acceleration. 
As a temperature of $1\,\mathrm{K}$ requires an acceleration of approximately $2.4 \times 10^{20}\mathrm{ms^{-2}}$ in~\eqref{eq: Unruh temp}, an observational verification of the prediction however lies well beyond current experimental capabilities. 
A comprehensive review of the Unruh effect and its applications is given in~\cite{Crispino:2007eb}.

Better prospects for observing the Unruh effect are provided by analogue spacetime systems~\cite{Unruh1981,Liberati,HeliumUniverse}, which simulate relativistic quantum fields in nonrelativistic hydrodynamical or condensed matter table-top experiments. In such systems, there are two main motivations to consider circular acceleration \cite{BellLeinaas, Bell:1986ir, Unruh:1998gq, Leinaas:1998tu} rather than linear acceleration. 
First, an experiment with circular acceleration can be run within a finite size laboratory for an arbitrarily long time. Second, unlike in a relativistic quantum field theory, a condensed matter system does not naturally incorporate a time dilation between the laboratory and the accelerated system; for circular motion, however, this time dilation factor is constant in time, and can therefore be accounted for easily at the data analysis stage \cite{Retzker:2007vql, Gooding:2020scc, Biermann:2020bjh, Unruh:2022gso, Gooding:2025tfp}. 
Further references on the theory of the circular motion Unruh effect are collected in~\cite{Parry:2025wub}. 

For circular motion, however, the Unruh effect cannot be characterised by the simple temperature formula~\eqref{eq: Unruh temp}, due to the lack of a `rotating vacuum' adapted to the circular motion \cite{Denardo:1978dj, Letaw:1979wy, Davies:1996ks, Earman:2011zz}. An effective temperature parameter can still be defined operationally in terms of the ratio of excitations and de-excitations of a localised quantum system~\cite{Unruh1976,DeWitt1979}, but this effective temperature depends not only on the proper acceleration, but also on all the parameters of the orbit and on the internal energy spacing of the detector. While the effective temperature broadly agrees with the linear acceleration prediction \eqref{eq: Unruh temp} over most of the parameter space \cite{Biermann:2020bjh, Unruh:1998gq, Good_2020, Parry:2024jrm}, a notable exception occurs in $2+1$ spacetime dimensions for a massless scalar field. In this case, the circular motion effective temperature is much smaller than the linear acceleration prediction when the internal energy spacing of the accelerating detector is small and the interaction with the quantum field lasts for a long time~\cite{Biermann:2020bjh}. The mathematical origin of this phenomenon is the weak falloff of the field's Wightman function along the circular worldline~\cite{Parry:2024jrm}. This case is especially important because a massless scalar field in $2+1$ spacetime dimensions is precisely the system that is simulated in the recent analogue spacetime proposals to observe the circular motion Unruh effect in 
Bose-Einstein condensates \cite{Gooding:2020scc,Gooding:2025tfp} and 
superfluid Helium thin-films~\cite{Bunney:2023ude}. 

It was shown in \cite{Parry:2025wub} that, for a massless scalar field in $2+1$ spacetime dimensions and a detector in circular motion, the effective temperature in the limit of small energy spacing and long interaction time can be raised to the order of magnitude of the linear acceleration prediction \eqref{eq: Unruh temp}. This is achieved by taking the small energy spacing and long interaction time limits simultaneously, in a controlled fashion,  \emph{provided\/} the detector-field coupling is allowed to change sign, again in a suitably controlled fashion. 
It was further noted that couplings of non-uniform sign arise naturally in certain entanglement-harvesting protocols in condensed matter systems~\cite{Lindel:2023rfi,Gooding:2023xxl}. These observations increase the parameter range in which an experimental verification of the circular motion Unruh effect can be pursued in condensed matter systems. 

What remained open in~\cite{Parry:2025wub}, however, was whether a sign change in the detector-field coupling is not merely \emph{sufficient}, under the constructions presented there, but also \emph{necessary}, under some reasonably general conditions, in order to raise the circular motion effective temperature to the order of the linear acceleration prediction~\eqref{eq: Unruh temp}. In the present paper, we provide a partial result in this direction: under a certain set of boundedness and localisation conditions on the coupling, we show that 
a sign change in the coupling is indeed necessary. We suspect that these boundedness and localisation conditions are not optimal, and we leave potential refinement of these conditions subject to future work. 

The rest of the paper is as follows. 
Section \ref{sec:field-detector-etc} gives a concise review of the transitions in a pointlike Unruh-DeWitt detector coupled linearly to a real massless scalar field in Minkowski spacetime, treated in first-order perturbation theory. The switch-on and switch-off of the interaction is handled in the asymptotically scaled switching family (ASSF) formalism of~\cite{Parry:2025wub}, and we recall the notion of an effective temperature that is defined in terms of the detector's excitation and de-excitation probabilities. 
Section \ref{sec: CM} specialises to uniform circular motion in $2+1$ dimensions, presenting the Small Frequency Suppression (SFS) criterion that determines when a nonzero effective temperature arises in the double limit of small gap and long time. Section \ref{sec:experimenter-instructions} shows how the SFS criterion can be satisfied by switchings that are not of uniform sign. Section \ref{sec: necessity} presents our new Theorem, showing how the SFS criterion necessarily leads to switchings of non-uniform sign, under certain boundedness and localisation assumptions. Section \ref{sec:conclusions} gives brief concluding remarks. 

We use units in which the speed of light, reduced Planck constant $\hbar$ and Boltzmann's constant $k_B$ are set to unity.  
In asymptotic formulae, 
$f(x)=O(g(x))$ denotes that $f(x)/g(x)$ is bounded in the limit of interest, and $f(x) = o(g(x))$ denotes that $f(x)/g(x)\to0$ in the limit of interest. 
The Fourier transform is defined as
\begin{equation}
    \widehat{f}(\omega) = \intinf \dd t \, \ee^{-\ii\omega t} f(t).
    \label{eq:fourier-def}
\end{equation}
The Heaviside theta function $\Theta(x)$ and the signum function $\sgn(x)$ are defined as 
\begin{equation}
    \Theta(x) = 
    \begin{cases}
        1 & \text{ for } x\geq0 \\
        0 & \text{ for } x<0,
    \end{cases}
\end{equation}
\begin{align}
    \sgn(x) = 
    \begin{cases}
        1 & \text{ for } x>0 \\
         0 & \text{ for } x=0\\
        -1 & \text{ for } x<0, 
    \end{cases}
\end{align}
and $\sinc$ denotes the analytic function defined by 
\begin{equation}
    \sinc z=\begin{cases} \frac{\sin z}{z} & z\neq 0\\ 1 & z=0.\end{cases}
\end{equation}

\section{Field, detector, long-time limit, and effective temperature} \label{sec:field-detector-etc}

In this section, we recall the basics of the field-detector model, the long-time limit, and the effective temperature in the framework in which we are working. A~more extended discussion is given in Section 2 of~\cite{Parry:2025wub}. 

\subsection{Field-detector model}
Our detector is a pointlike two-level Unruh-DeWitt detector~\cite{Unruh1976, DeWitt1979}, following a worldline $\xx(\tau)$ parametrised by proper time $\tau$ in Minkowski spacetime, which to begin with we assume to have dimension $d\geq 3$. The detector Hilbert space is $\mathcal{H}_D \cong \mathbb{C}^2$, with orthonormal energy eigenbasis $\{\ket{0},\ket{E}\}$, such that $H_D\ket{0} = 0$ and $H_D\ket{E}= E\ket{E}$, where $E\ne0$ is the detector's energy gap. When $E>0$, $\ket{0}$ is the ground state and $\ket{E}$ is the excited state; when  $E<0$, the roles of the states are reversed. 

The detector interacts with a real massless scalar field~$\phi$. The combined system has Hilbert space $\mathcal{H}_D\otimes \mathcal{H}_\phi$, where $\mathcal{H}_\phi$ is the field's Fock space induced by the Minkowski vacuum. Working in the interaction picture, we take the interaction Hamiltonian to be 
\begin{equation}
    H_\mathrm{int}(\tau) = c \chi(\tau) \mu(\tau) \otimes \phi\!\left(\xx(\tau) \right),
\end{equation}
where $c>0$ is the coupling constant, $\chi(\tau)$ is a real-valued switching function that determines how the interaction is turned on and off, and $\mu(\tau)$ is the detector's monopole moment operator. The differentiability and falloff properties of $\chi$ will have a key role in the long time limits that we shall discuss. 

Before the interaction is turned on, the detector is prepared in the state $\ket{0} \in \mathcal{H}_D$ and the field is prepared in the state $\ket{\Psi} \in \mathcal{H}_\phi$. In first order perturbation theory in the coupling constant $c$, the probability to find the detector in the state $\ket{E}$ after the interaction is turned off, without measuring the field state, is proportional to the response function, given by 
\begin{equation} \label{eq:response}
    \F_\chi(E) = \intinf \dd \tau \intinf \dd \tau' \, \chi(\tau) \chi(\tau') \, \ee^{-\ii E (\tau-\tau')} \W(\tau,\tau'),
\end{equation}
where $\W(\tau,\tau') := \expval{\phi(\xx(\tau))\phi(\xx(\tau'))}{\Psi}$ is the pullback of the field's Wightman function in the state $\ket{\Psi}$ to the detector's worldline \cite{Unruh1976, DeWitt1979, birrell, Wald:1995yp, Junker:2001gx}. The response function is well defined whenever the field's Wightman function is a distribution of Hadamard type \cite{Radzikowski:1996pa,Decanini:2005eg} and $\chi$ is sufficiently differentiable and has sufficient falloff. 
The constant of proportionality is $c^2 {| \! \bra{E} \mu(0) \ket{0} \! |}^2$, which is independent of the detector's worldline and~$\ket{\Psi}$, and it is also independent of $E$ when $\mu(\tau) = \ee^{-\ii E \tau} \ket{0} \! \bra{E} + \ee^{\ii E \tau} \ket{E} \! \bra{0}$, which is the usual field operator of a two-level system. We suppress this constant of proportionality in what follows.

We further assume the system to be stationary, in the sense that the detector's worldline is an integral curve of a Killing vector field that is timelike at least in the neighbourhood of the worldline \cite{Kottler1912,Synge1967,Letaw:1980yv, Russo:2009yd, Bunney:2023mkh, Fewster:2023zrw,fewster2026quantumenergyinequalitiestwist}, and that the Wightman function is invariant under the Poincar\'e transformations generated by this Killing vector. $\W(\tau,\tau')$ then depends on $\tau$ and $\tau'$ only through their 
difference, and we may write $\W(\tau-\tau'):=\W(\tau,\tau')$. 
$\F_\chi$~\eqref{eq:response} then takes the form \cite{Fewster:2016ewy}
\begin{equation} \label{eq:response convolution}
    \F_\chi(E) = \frac{1}{2\pi}\intinf \dd \omega \, |\widehat{\chi}(\omega)|^2 \widehat{\W}(\omega + E).
\end{equation}
While $\widehat{\W}$ is in general a distribution, we assume from now on that it is a bounded function, and continuous except possibly at zero argument. 
This will be the case for $2+1$ circular motion in Section~\ref{sec: CM}.

\subsection{Long time limit} \label{sec: long time limit}

\subsubsection{Switching function family} \label{subsubsec:wishlist}

To formalise the limit in which the interaction lasts for a long time, we introduce a one-parameter family of switching functions $\chi_\lambda$ with $\lambda >0$, where $\lambda$ determines the interaction duration. As $\lambda$ grows, we wish the interaction to be strongly  supported on an interval of length proportional to~$\lambda$, and the coupling strength to remain approximately constant in some averaged sense. We then expect $\F_{\chi_\lambda}$ \eqref{eq:response convolution} to grow linearly with~$\lambda$,
so that the quantity of interest is the scaled response function 
\begin{equation} \label{eq:Flambda}
    \F_\lambda(E) := \lambda^{-1}\F_{{\chi_{{}_\lambda}}}(E), 
\end{equation}
and we expect $\F_\lambda$ to have a finite limit as $\lambda\to\infty$. There exist families $\chi_\lambda$ for which these expectations are realised, with the outcome that  
\begin{align}
\F_\lambda(E) \xrightarrow[\lambda\to\infty]{} \alpha \widehat{\W}(E)\,,
\label{eq:Flambda-pointwise-limit}
\end{align}
for some positive constant~$\alpha$, where the $\lambda\to\infty$ limit in \eqref{eq:Flambda-pointwise-limit} is understood pointwise in $E$~\cite{Fewster:2016ewy, Parry:2025wub}. This is consistent with the usual informal long interaction duration limit, where one factors out the total interaction duration from the response function and interprets the remaining quantity as the transition rate~\cite{Unruh1976,DeWitt1979,birrell}. 

For the $2+1$ circular motion Unruh effect, however, we need to consider the simultaneous limit of $\lambda\to\infty$ and $E\to0$. To this end, we work with switching function families that provide a precise control of both the switching function and its Fourier transform, as introduced in~\cite{Parry:2025wub}.

\subsubsection{Asymptotically scaled switching family}

A prototype of a switching function family that satisfies the wish list of Section \ref{subsubsec:wishlist} is the adiabatic switching, $\chi_\lambda(\tau) = \chi(\tau/\lambda)$, where the function $\chi$ satisfies suitable regularity and falloff conditions~\cite{Parry:2025wub,Fewster:2016ewy}: in this case, the constant $\alpha$ in \eqref{eq:Flambda-pointwise-limit} is equal to~$\Vert \chi \Vert^2$. We work with the  generalisation introduced in~\cite{Parry:2025wub}, called an asymptotically scaled switching family (ASSF). The generalisation relies on the following proposition, which is Proposition 3.1 in~\cite{Parry:2025wub}, where it is proved.

\begin{proposition}\label{prop:chilambda}
    Let $\chi_\lambda\in C^1(\mathbb{R})$ $(\lambda>0)$ be a family of absolutely integrable switching functions
    with the property that $\lambda^{-1}\widehat{\chi}_\lambda(u/\lambda)$ converges almost everywhere in $u$ as $\lambda\to\infty$. Suppose that
    there exists $\eta\in L^2(\mathbb{R},\dd u/(2\pi))$ with the property that, for all sufficiently large $\lambda>0$, one has $|\widehat{\chi}_\lambda(\omega)|\le \lambda \eta(\lambda\omega)$ for almost all $\omega\in\mathbb{R}$. 
    Then there is $\xi\in L^2(\mathbb{R},\dd t)$ so that 
    \begin{equation}\label{eq:xi}
        \widehat{\xi}(u) = \lim_{\lambda\to\infty} \frac{\widehat{\chi}_\lambda(u/\lambda)}{\lambda}
    \end{equation}
    for almost all $u\in\mathbb{R}$.
    Furthermore, if $V:\mathbb{R}^2\to \mathbb{C}$ is any bounded measurable function that is continuous at 
    $(E_*,0)\in \mathbb{R}^2$, 
    then one has a double limit
\begin{equation}\label{eq:Vlimit}
    \frac{1}{2\pi} \intinf \dd\omega\, \frac{|\widehat{\chi}_\lambda(\omega)|^2}{\lambda} \, V(E,\omega) \xrightarrow[E\to E_*]{\lambda\to\infty} \|\xi\|^2 V(E_*,0).
\end{equation}
In particular, the scaled response function $\F_\lambda$ \eqref{eq:Flambda} obeys
    \begin{equation}
    \label{eq:F-to-xi2What}
       \F_\lambda(E) \xrightarrow[E\to E_*]{\lambda\to\infty}  {\Vert \xi \Vert}^2 \, \widehat{\W}(E_*), 
    \end{equation} 
    for every $E_*\in\mathbb{R}$ at which $\widehat{\W}$ is continuous. This implies that one also has
    \begin{equation}
    \label{eq:F-to-xi2What-fixedE}
    \F_\lambda(E_*)
    \xrightarrow[\lambda\to\infty]{}
    {\Vert \xi \Vert}^2 \, \widehat{\W}(E_*)
    \end{equation}
    for all such~$E_*$.
\end{proposition}

Given Proposition~\ref{prop:chilambda}, the ASSF is now defined as follows. 

\begin{definition}\label{def:assf}
(Asymptotically Scaled Switching Family.) 
If a switching function family $\chi_\lambda$ satisfies the assumptions of Proposition \ref{prop:chilambda} and the function 
$\xi \in L^2(\mathbb{R},\dd t)$ provided by Proposition \ref{prop:chilambda} is not the zero element of $L^2(\mathbb{R},\dd t)$, we call $\chi_\lambda$ an Asymptotically Scaled Switching Family (ASSF). 
\end{definition}

We remark that ASSF switching functions need not have compact support, nor need they have uniform sign. The latter property will be important in Sections \ref{sec: CM} and~\ref{sec: necessity}. 

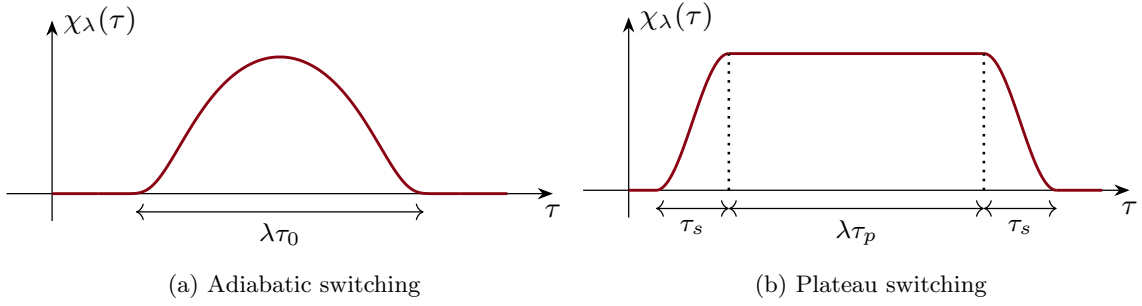
\begin{figure}[t]
\centering

\begin{subfigure}[t]{0.5\textwidth}
    \centering
    \resizebox{\linewidth}{!}{%
\begin{tikzpicture}

\tikzset{contour/.style={draw=curve, line width=0.9pt}}

\def\supp{2}
\def\height{1.5}

\draw[-Stealth] (-3,0) -- (3,0) node[below, font =\small]{$\tau$};
\draw[-Stealth] (-2.5,-0.3) -- (-2.5,\height+0.4) node[right, font=\small]{$\chi_\lambda(\tau) $};


\draw[contour, smooth, samples=200, domain=-\supp+0.01:\supp-0.01]
plot (\x,{\height*exp(1 - 1/(1-\x*\x/(\supp*\supp))^2)});

\draw[contour] (-2.5,0) -- (-\supp+0.01,0);
\draw[contour] (\supp-0.01,0) -- (2.5,0);

\draw[<->] (-\supp+0.42,-0.2) -- (\supp-0.42,-0.2) node[below, font=\footnotesize] at (0,-0.2) {$\lambda \tau_0$};
    
\end{tikzpicture}
}
\caption{Adiabatic switching}
\end{subfigure}%
\hfill
\begin{subfigure}[t]{0.5\textwidth}
\centering
\resizebox{\linewidth}{!}{%
\begin{tikzpicture}

\tikzset{contour/.style={draw=curve, line width=0.9pt}}

\def\taus{0.8}
\def\taup{2.8}
\def\height{1.5}
\def\taustart{-2.2}

\draw[-Stealth] (-3,0) -- (3,0) node[below, font =\small]{$\tau$};
\draw[-Stealth] (-2.5,-0.3) -- (-2.5,\height+0.4) node[right, font=\small]{$\chi_\lambda(\tau)$};

\draw[contour, smooth, samples=200, domain=\taustart:{\taustart+\taus}]
plot (\x,{\height*0.5*(1-cos(180*(\x-\taustart)/\taus))});

\draw[contour] ({\taustart+\taus},\height) -- ({\taustart+\taus+\taup},\height);

\draw[contour, smooth, samples=200, domain={\taustart+\taus+\taup}:{\taustart+2*\taus+\taup}]
plot (\x,{\height*0.5*(1+cos(180*(\x-\taustart-\taus-\taup)/\taus))});

\draw[contour] (-2.5,0) -- (\taustart,0);

\draw[contour] ({\taustart+2*\taus+\taup},0) -- (2.7,0);

\draw[dotted, thick] ({\taustart+\taus},0) -- ({\taustart+\taus},\height);

\draw[dotted, thick] ({\taustart+\taus+\taup},0) -- ({\taustart+\taus+\taup},\height);

\draw[<->] ({\taustart+\taus},-0.2) -- ({\taustart+\taus+\taup},-0.2)
node[below, font=\footnotesize] at ({\taustart+\taus+\taup/2},-0.2) {$\lambda \tau_p$};

\draw[<->] (\taustart,-0.2) -- ({\taustart+\taus},-0.2)
node[below, font=\footnotesize] at ({\taustart+\taus/2},-0.2) {$\tau_s$};

\draw[<->] ({\taustart+\taus+\taup},-0.2) -- ({\taustart+2*\taus+\taup},-0.2)
node[below, font=\footnotesize] at ({\taustart+\taus+\taup+\taus/2},-0.2) {$\tau_s$};

\end{tikzpicture}
}
\caption{Plateau switching}    
\end{subfigure}

\caption{Panel (a) shows the graph of an adiabatic switching $\chi_\lambda$~\eqref{eq: adiabatic switching}: as $\lambda$ increases, the whole profile of $\chi_\lambda$ is stretched proportionally to~$\lambda$.
Panel (b) shows the graph of a plateau switching $\chi_\lambda$~\eqref{eq: plateau switching}: as $\lambda$ increases, the constant section of $\chi_\lambda$ is stretched proportionally to~$\lambda$, but the initial and final switch-on and switch-off intervals have fixed duration. }
\label{fig:adiabatic-and-plateau}

\end{figure}

An example of an ASSF is the adiabatic switching, 
\begin{equation} \label{eq: adiabatic switching}
    \chi_\lambda(\tau)=\chi(\tau/\lambda),
\end{equation}
where $\chi\in C_0^1(\mathbb{R})$ is not identically vanishing. This is illustrated in Figure \ref{fig:adiabatic-and-plateau} Panel~(a). 
The conditions of Proposition \ref{prop:chilambda} hold with $\eta = |\widehat{\chi}|$, and the proposition provides $\xi=\chi$~\cite{Parry:2025wub}. Therefore, \eqref{eq:F-to-xi2What-fixedE} gives \eqref{eq:Flambda-pointwise-limit} with $\alpha = \Vert \chi \Vert^2$, 
for every $E\in \mathbb{R}$ at which $\widehat{\W}$ is continuous. 

Another example of an ASSF is the plateau switching, 
\begin{equation} \label{eq: plateau switching}
    \chi_\lambda(\tau) = \int_{-\infty}^\tau \dd \tau' \, \big(\psi(\tau') - \psi(\tau'-\tau_s-\lambda \tau_p) \big),
\end{equation}
where $\tau_s$ and $\tau_p$ are positive constants, and $\psi \in C_0(\mathbb{R})$. $\psi$~is assumed non-negative, not identically vanishing, and with support contained in an interval of duration~$\tau_s$. 
This is illustrated in Figure \ref{fig:adiabatic-and-plateau} Panel~(b). 
The conditions of Proposition \ref{prop:chilambda} hold, and $\xi$ is $\widehat{\psi}(0)$ times the characteristic function of $[0,\tau_p]$, where $\widehat{\psi}(0)>0$ by the assumptions on~$\psi$~\cite{Parry:2025wub}. Therefore, \eqref{eq:F-to-xi2What-fixedE} gives \eqref{eq:Flambda-pointwise-limit} with $\alpha = \tau_p |\widehat{\psi}(0)|^2$, for every $E\in \mathbb{R}$ at which $\widehat{\W}$ is continuous.

\subsection{Effective temperature}
We now describe a notion of an effective temperature that characterises the detector's response in the long time limit, implemented with an ASSF.

In the open quantum systems framework, where the field is modelled as an ambient Markovian environment and the detector-field coupling is so weak that the field state remains effectively unchanged, the detector state approaches at late times the Gibbs-like density matrix \cite{DeBievre:2006pys, Juarez-Aubry:2019gjw}
\begin{equation} \label{eq: Gibbs-like state}
    \rho(E) = \frac{1}{1+\ee^{-E/T(E)}} 
    \begin{pmatrix}
        1 & 0 \\ 0 & \ee^{-E/T(E)}
    \end{pmatrix},
\end{equation}
where
\begin{equation} \label{eq:detailed balance temp}
    T(E) = \frac{E}{\log\!\left(\displaystyle{\frac{\widehat{\W}(-E)}{\widehat{\W}(E)}} \right)}.
\end{equation}
When the field is prepared in the Minkowski vacuum, the only non-inertial stationary trajectories for which $T(E)$ is independent of $E$ are the uniformly linearly accelerated trajectories with proper acceleration $a>0$, 
giving $T(E) = a/(2\pi)$. 
$\rho(E)$~\eqref{eq: Gibbs-like state} is then a genuine thermal Gibbs state at the Unruh temperature~$a/(2\pi)$, 
and $T(E)$ is the corresponding detailed balance temperature~\cite{Einstein, terhaar-book}. For other types of non-inertial stationary motion, $T(E)$ generally depends on~$E$, and $\rho(E)$ \eqref{eq: Gibbs-like state} is not strictly thermal. However, on intervals of $E$ where this $E$-dependence is weak, \eqref{eq: Gibbs-like state} may still be viewed as approximately thermal at the gap-dependent effective temperature~$T(E)$. For studies of $T(E)$ in stationary detector settings, see \cite{Biermann:2020bjh, Good_2020, Bunney:2023vyj, Bunney:2023ipk, Parry:2024jrm, Passegger:2025amw}.

In the present setting, where the long time limit is implemented by an ASSF $\chi_\lambda$ as $\lambda \to \infty$, we define the $\lambda$-dependent detailed balance temperature \cite{Parry:2025wub}
\begin{equation} \label{eq: Tlambda(E)}
    T_\lambda(E) = \frac{E}{\log\!\left(\displaystyle{\frac{\F_\lambda(-E)}{\F_\lambda(E)}} \right)},
\end{equation}
whenever the right-hand side is defined, where $\F_\lambda(E)$ is given in \eqref{eq:Flambda}. In the long time limit $\lambda \to \infty$, Proposition \ref{prop:chilambda} gives $\F_\lambda(E) \to \Vert\xi\Vert^2 \widehat{\W}(E)$ for all $E\in \mathbb{R}$ for which $\widehat{\W}$ is continuous. Hence, in this limit, $T_\lambda(E)\to T(E)$, with $T(E)$ as in \eqref{eq:detailed balance temp}. We therefore refer to $T_\lambda(E)$ as the (finite time) detailed balance temperature, and we use it to characterise both the response of the detector with switching function $\chi_\lambda$ and the limit of the response when $\lambda \to \infty$.

\section{Circular motion in \texorpdfstring{$2+1$}{2+1} dimensions} \label{sec: CM}

We now specialise to a detector in uniform circular motion in $(2+1)$ Minkowski spacetime, with the field prepared in the Minkowski vacuum.

\subsection{Response}
Working in a set of standard Minkowski coordinates, 
$\xx = (t,x,y)$ with $\dd s^2 = - \dd t^2 + \dd x^2 + \dd y^2$, 
we may choose a Lorentz frame in which the circular motion worldline reads
\begin{equation} \label{eq: circular worldline}
    \xx(\tau) = \left( \gamma \tau, R \cos\!\left(\frac{\gamma v \tau}{R} \right), R \sin\!\left(\frac{\gamma v \tau}{R} \right)\right),
\end{equation}
where $R>0$ is the orbital radius, $v$ is the orbital speed with $0<v<1$, and $\gamma=(1-v^2)^{-1/2}$ is the Lorentz factor. The proper acceleration is 
\begin{align}
a=\frac{\gamma^2v^2}{R} \,. 
\label{eq:circ-properacceleration}
\end{align}
The worldline \eqref{eq: circular worldline} is stationary,  and as the Minkowski vacuum is Poincar\'e invariant, the response function is given by \eqref{eq:response convolution}, where \cite{Biermann:2020bjh}
\begin{equation} \label{eq: W hat CM}
    \widehat{\W}(E) = \frac{1}{4} - \frac{k E}{2\pi \gamma} \intoinf \dd z \, \frac{\sinc(kEz)}{\sqrt{1-v^2\sinc^2\!z}},
\end{equation}
with $k=2R/(\gamma v)$.

As observed in~\cite{Parry:2025wub}, $\widehat{\W}$ \eqref{eq: W hat CM} satisfies $0\leq \widehat{\W} \leq \frac{1}{2}$, and 
$\widehat{\W}(E)$ is continuous everywhere except at $E=0$, where it has well-defined one-sided limits. The discontinuity at $E=0$ is seen by splitting $\widehat{\W}$ as \cite{Biermann:2020bjh, Parry:2024jrm, Parry:2025wub}
\begin{equation} \label{eq: W hat CM alternative split}
    \widehat{\W}(E) = \frac{1}{4}-\frac{k}{2\pi \gamma} U(E) - \frac{\sgn(E)}{4\gamma},
\end{equation}
where 
\begin{equation}
    U(E) =  E\intoinf \dd z \,\sinc(kE z)\left(\frac{1}{\sqrt{1-v^2 \sinc^2\!z}}-1 \right) ,  
\label{eq:U-def}
\end{equation}
since a dominated convergence argument shows that $U$ \eqref{eq:U-def} is continuous everywhere, including at $E=0$, with $U(E) = O(E)$ as $E\to 0$. 

Inserting \eqref{eq: W hat CM alternative split} into \eqref{eq:response convolution} and \eqref{eq:Flambda} gives 
\begin{equation} \label{eq: Flambda chilambda}
    \F_\lambda(E) = \frac{1}{4\lambda}\Vert \chi_\lambda \Vert^2 - \frac{k E}{4\pi^2 \gamma}\intinf \dd \omega \,  \frac{|\widehat{\chi}_\lambda(\omega)|^2}{\lambda} V(E,\omega) - \frac{\sgn(E)}{8\pi \gamma \lambda} \int_{-|E|}^{|E|} \dd \omega \, |\widehat{\chi}_\lambda(\omega)|^2,
\end{equation}
where 
\begin{align}
    V(E,\omega) &=\frac{U(E+\omega) + U(E-\omega)}{2E}  
= \intoinf \dd z \, \sinc(kEz)\cos(k\omega z) \!\left(\frac{1}{\sqrt{1-v^2 \sinc^2\!z}}-1 \right).
\label{eq:Vfunc-def}
\end{align}
The function $V$ \eqref{eq:Vfunc-def} is continuous and bounded on $\mathbb{R}^2$, and it is even in both arguments separately.
This shows that the second and third terms \eqref{eq: Flambda chilambda} are odd in~$E$, while the first term is independent of $E$. These properties make \eqref{eq: Flambda chilambda} convenient for analysing $\F_\lambda(E)$ in the double limit of $E\to0$ and $\lambda\to\infty$.

\subsection{Small gap and long time limits in succession: vanishing small gap temperature}

We now turn to the effective temperature $T_\lambda(E)$ \eqref{eq: Tlambda(E)} in the limit of small gap and long time, $E\to 0$ and $\lambda\to\infty$. We first consider the limits taken in succession. 

Suppose the long time limit is taken first, and suppose that the ASSF family is such that the outcome is~\eqref{eq:Flambda-pointwise-limit}, with $\widehat{\W}$ given by~\eqref{eq: W hat CM alternative split}: an example of assumptions that guarantee this was given in~\cite{Parry:2025wub}. 
At small gap, we then find \cite{Biermann:2020bjh}
\begin{equation}
    \widehat{\W}(E) = \frac{\gamma - \sgn(E)}{4\gamma} + O(E),
\end{equation}
by which \eqref{eq:detailed balance temp} gives 
\begin{equation}
    T(E) = \frac{|E|}{\log\!\left(\displaystyle{\frac{\gamma+1}{\gamma-1}} \right)} + O(E^2), 
\end{equation}
which vanishes linearly in $E$ as $E\to0$. The small gap effective temperature is therefore significantly lower than expected from the linear acceleration prediction \eqref{eq: Unruh temp} with the circular proper acceleration~\eqref{eq:circ-properacceleration}~\cite{Biermann:2020bjh}. The underlying reason for this discrepancy is the discontinuity of $\widehat{\W}$ at $E=0$, which is characteristic of circular motion in $2+1$ dimensions~\cite{Parry:2024jrm}.  

Conversely, suppose the small gap limit is taken first. 
When the ASSF family is non-negative and satisfies modest technical boundedness assumptions, it was shown in \cite{Parry:2025wub} that the small gap effective temperature tends to zero as $O(\lambda^{-1})$ as $\lambda \to \infty$. Again, this is significantly lower than expected from the linear acceleration prediction \eqref{eq: Unruh temp} with the circular proper acceleration~\eqref{eq:circ-properacceleration}.

\subsection{Simultaneous small gap and long time limit: recovery of positive small gap temperature}

We wish to recover a nonzero effective temperature in the small gap and long time limit. To this end, we now consider the limits not in succession, but simultaneously, in a sense to be established. 

We assume that $\lambda$ and $E$ are related by the positive-valued even function $\lambda(E)$, such that $\lambda(E)\to \infty$ as $E\to 0$. As $\chi_\lambda$ is an ASSF by assumption, $\chi_\lambda$ is absolutely integrable and in $C^1$ by Proposition~\ref{prop:chilambda}. Further, by Proposition \ref{prop:chilambda}, $\chi_\lambda$ is square-integrable for sufficiently large $\lambda$, and $\lambda^{-1}\Vert\chi_\lambda \Vert^2 \to \Vert \xi \Vert^2>0$ as $\lambda \to \infty$.

Consider $\F_\lambda(E)$ \eqref{eq: Flambda chilambda} 
as $\lambda\to\infty$. The first term in \eqref{eq: Flambda chilambda} converges to $\Vert\xi\Vert^2/4$, and the integral in the second term has a non-zero double limit by Proposition~\ref{prop:chilambda}. We therefore have 
\begin{equation}
    \F_{\lambda(E)}(E) = \frac{\Vert \xi \Vert^2}{4}(1+o_e(1)) - \frac{kE}{2\pi \gamma}\Vert \xi \Vert^2 V(0,0) (1+o_e(1)) - \frac{\sgn(E)}{8\pi \gamma \lambda(E)} \int_{-|E|}^{|E|} \dd \omega |\widehat{\chi}_{\lambda(E)}(\omega)|^2 
\label{eq: Flambda chilambda half}
\end{equation}
as $E\to 0$, where the subscript $e$ on the $o_e(1)$ terms indicates that these terms are even in~$E$. 

The third term in \eqref{eq: Flambda chilambda half} depends on the magnitude of $\lambda^{-1}|\widehat{\chi}_{\lambda(E)}(\omega)|^2$ over the interval $|\omega|\leq |E|$ as $E\to 0$, and this is sensitive to the precise relationship~$\lambda(E)$. We now show how a positive temperature can be recovered in the combined long-time-small-gap limit by controlling this third term.

It is clear from \eqref{eq: Tlambda(E)} that a finite positive temperature in the small gap limit occurs if and only if 
\begin{equation} \label{eq:positive temp condition}
    \frac{\F_{\lambda(E)}(-E)}{\F_{\lambda(E)}(E)} = 1 +\frac{E}{T_0} + o(E)
\end{equation}
as $E\to 0$, where $T_0$ is a positive constant, and we then have $T_{\lambda(E)}(E) = T_0 + o(1)$ as $E\to 0$. The crucial issue is therefore the balance of the even and odd parts of $\F_{\lambda(E)}(E)$. This balance was analysed in \cite{Parry:2025wub}, with the outcome that 
a positive temperature $T_0$ is attained if and only if
\begin{equation}
    \frac{1}{|E|\lambda(E)}\int_{-|E|}^{|E|}\dd \omega \, |\widehat{\chi}_{\lambda(E)}(\omega)|^2
\label{eq:pre-sfs}
\end{equation}
has a finite, possibly zero, limit as $E\to 0$. A non-zero limit would modify the limiting temperature~$T_0$, so we consider only the case where the limit of \eqref{eq:pre-sfs} vanishes. This condition is 
\begin{equation}\label{eq:SFS}
    \frac{1}{|E|\lambda(E)} \int_{-|E|}^{|E|}\dd \omega \, |\widehat{\chi}_{\lambda(E)}(\omega)|^2 = o(1) 
\qquad \text{(SFS)}
\end{equation}
as $E\to 0$. We refer to \eqref{eq:SFS} as the Small Frequency Suppression (SFS) condition. When SFS holds, one then finds 
\begin{equation} \label{eq:positive temp}
    T_{\lambda(E)}(E) = \frac{a}{2\pi I(v)}\big(1 + o(1) \big),
\end{equation}
as $E\to 0$, where 
\begin{equation}\label{eq:I(v)}
    I(v) := \frac{4v}{\pi^2}V(0,0) =\frac{4v}{\pi^2} \intoinf \dd z \left( \frac{1}{\sqrt{1-v^2\sinc^2\!z}}-1 \right).
\end{equation}

$I(v)$ \eqref{eq:I(v)} vanishes cubically as $v\to0$ and diverges logarithmically as $v\to1$, but it is of order unity for intermediate values of~$v$. 
The small gap temperature \eqref{eq:positive temp} is hence comparable with the linear acceleration prediction, $a/(2\pi)$, over most of the range of~$v$. 

\section{Instructions for an experimenter in the small gap regime} \label{sec:experimenter-instructions}

We have established that a positive small gap temperature is recovered in the simultaneous small gap and long time limit when the SFS condition \eqref{eq:SFS} holds, and this temperature is comparable to the linear acceleration prediction for most of the velocity range. However, in operational terms, what does the SFS condition imply about the ASSF family $\chi_\lambda$ and the relation between $\lambda$ and~$E$?  
How should an experimenter turn the interaction on and off if their experiment is operating in the small gap regime and they wish to probe the effective temperature prediction~\eqref{eq:positive temp}?

\subsection{Adiabatic switching} \label{sec:adiabatic-smallgap}

Consider the adiabatic switching family~\eqref{eq: adiabatic switching}. In this case, the left-hand side of \eqref{eq:SFS} becomes
\begin{equation} \label{eq:adiabatic SFS}
   \frac{1}{|E|\lambda(E)} \int_{-|E|}^{|E|}\dd \omega \, |\widehat{\chi}_{\lambda(E)}(\omega)|^2 = \frac{1}{|E|}\int_{-r(E)}^{r(E)} \dd u\, |\widehat{\chi}(u)|^2,
\end{equation}
where we have changed variables to $u=\lambda(E)\omega$ and defined $r(E):=|E|\lambda(E)$. The small gap behaviour of \eqref{eq:adiabatic SFS} depends on the behaviour of $r(E)$ as $E\to0$. Following the analysis of Appendix A of \cite{Parry:2025wub}, one finds
\begin{equation} \label{eq: adiabatic SFS asymptotic}
    \frac{1}{|E|}\int_{-r(E)}^{r(E)} \dd u\, |\widehat{\chi}(u)|^2 \sim 
    \begin{cases}
        \displaystyle{2\pi \Vert \chi \Vert^2 |E|^{-1}} & \text{ if $r(E)\to \infty$} \\[1ex]
        \displaystyle{|E|^{-1} \int_{-S}^S \dd u \, |\widehat{\chi}(u)|^2} & \text{ if $r(E)\to S>0$} \\[2.6ex]
        \displaystyle{2 |\widehat{\chi}(0)|^2 \lambda(E)} & \text{ if $r(E)\to 0$},
    \end{cases}
\end{equation}
as $E\to 0$, where the last line assumes $\widehat{\chi}(0)\ne0$, and the middle line assumes that $\widehat{\chi}$ is not identically vanishing on the interval $[-S,S]$. 

From \eqref{eq: adiabatic SFS asymptotic}, we see that if $\widehat{\chi}(0)\neq 0$, the SFS condition \eqref{eq:SFS} cannot hold, for any $\lambda(E)$, given our standing assumption that $\lambda(E)\to\infty$ as $E\to0$. 

However, if $\widehat{\chi}(0)=0$, it is seen directly from the rightmost expression in \eqref{eq:adiabatic SFS} that when $r(E)\to 0$ as $E\to0$, SFS is satisfied if 
$\widehat{\chi}(\omega)\to 0$ sufficiently rapidly as $\omega \to 0$. 
For example, if $\lambda(E) \propto |E|^{-\alpha}$ with $0<\alpha<1$ and $\widehat{\chi}(\omega) = O(|\omega|^q)$ as $\omega \to 0$ with $q>0$, then SFS is satisfied when $q>\frac{\alpha}{2(1-\alpha)}$. If $\widehat{\chi}$ vanishes not just at the origin but on some interval around the origin, SFS can even be satisfied for $r(E)$ that is bounded as $E\to0$ but does not necessarily tend to zero. 

For SFS to hold with adiabatic switching, it is hence necessary that $\widehat{\chi}(0)=0$. 
As ${\widehat{\chi}(0) = \intinf \dd \tau \, \chi(\tau)}$, this implies that $\intinf \dd \tau \, \chi_{\lambda}(\tau)=0$ for every $\lambda>0$. The switching function thus cannot have a uniform sign: it must take both positive and negative values and average to zero. Examples of adiabatic switching families that satisfy SFS were given in~\cite{Parry:2025wub}, including compactly-supported profiles, profiles with power-law decay, and profiles with Gaussian decay. 

\subsection{Near-adiabatic switching}

It was shown in \cite{Parry:2025wub} that for every adiabatic switching family satisfying SFS, with suitably strong differentiability and falloff conditions, a compactly-supported asymptotically adiabatic switching family satisfying SFS can be constructed by a suitable cropping of the past and future tails. The cropped switchings do not necessarily satisfy $\widehat{\chi}_\lambda(0)=0$, but they have the property that $\widehat{\chi}_\lambda(0)\to0$ as $\lambda\to\infty$, and they do not have fixed sign.

\subsection{Plateau switching}

Now consider the plateau switching family~\eqref{eq: plateau switching}. It is shown in \cite{Parry:2025wub} that SFS is not satisfied, because of the assumption that $\psi$ is non-negative and not identically zero, which implies $\widehat{\psi}(0)>0$. Relaxing the assumptions to allow $\widehat{\psi}(0)=0$ would mean that $\chi_\lambda$ vanishes on the constant section that is stretched proportionally to~$\lambda$, and this would not qualify as an implementation of an interaction that lasts for a long time.

\section{The necessity to change sign} \label{sec: necessity}

Within adiabatic switching families, we have seen that SFS \eqref{eq:SFS} cannot be satisfied if the functions have a definite sign. The near-adiabatic SFS switching families constructed in \cite{Parry:2025wub} also have the property that the functions do not have a fixed sign. No examples of fixed sign ASSF families satisfying SFS are known to us. We now show in Theorem~\ref{prop:necessary_sign_changes} that fixed sign ASSF families satisfying SFS do not exist under certain boundedness and localisation assumptions. Theorem \ref{prop:necessary_sign_changes} is the substantial new contribution of the present paper. 

\begin{theorem} \label{prop:necessary_sign_changes}
Let the switching function family $\chi_\lambda$ ($\lambda>0$) be an ASSF\null. 
Assume that for sufficiently large $\lambda$ the following three properties hold: 
\begin{itemize}
    \item[(i)] $|\chi_\lambda| \leq M \lambda^{1/2}$, 
    where the constant $M>0$ is independent of~$\lambda$;
    \item[(ii)] 
    $\tau \chi_\lambda(\tau)$ is absolutely integrable, and 
\begin{align}
        \int_{-\infty}^\infty \mathrm{d} \tau \, |\tau| |\chi_\lambda(\tau)| \leq C \lambda \int_{-\infty}^\infty \mathrm{d} \tau 
        \, |\chi_\lambda(\tau)| \,,
\label{eq:localisation-condition}
\end{align}
where the constant $C>0$ is independent of~$\lambda$; 
    \item[(iii)]
    $\chi_{\lambda}\ge0$. 
\end{itemize}
Then the SFS condition \eqref{eq:SFS} does not hold as $E\to 0$, for any scaling $\lambda = \lambda(E)$ for which  $\lambda(E) \to \infty$  and $|E|\lambda(E) \to 0$ as $E\to 0$. 
\end{theorem}

In the assumptions of Theorem~\ref{prop:necessary_sign_changes}, property (i) bounds the growth of the supremum of $|\chi_{\lambda}|$ as $\lambda\to\infty$, by a $\lambda^{1/2}$ power law. Note in particular that any bounded ASSF family satisfies~(i). 

Property (ii) states that the functions $\chi_{\lambda}$ are sufficiently localised. 
For example, property (ii) holds if the support of $\chi_\lambda$ is contained in $[-C\lambda, C\lambda]$, where $C$ is a $\lambda$-independent positive constant; the same $C$ then appears in the bound~\eqref{eq:localisation-condition}. 
Property (ii) holds also for the adiabatic ASSF family, $\chi_\lambda(\tau) = \chi(\tau/\lambda)$, where $\chi \in C^1(\mathbb{R})$ is absolutely integrable, provided $\tau\chi(\tau)$ is absolutely integrable, and this latter property then supplies the requisite sense of localisation. More generally, one can read (ii) as a statement that the spread of the probability distribution $|\chi_\lambda(\tau)|/\|\chi_\lambda\|_1$, measured by its first absolute moment, grows no faster than linearly in~$\lambda$.

Property (iii) is the key assumption of uniform sign of the switching function. 

Finally, the scaling $\lambda = \lambda(E)$ is assumed to satisfy $\lambda(E) \to \infty$ and $|E|\lambda(E) \to 0$ as $E\to 0$. The assumption $\lambda(E) \to \infty$ as $E\to 0$ is the long time limit. The assumption $|E|\lambda(E) \to 0$ as $E\to 0$ is introduced because 
it leads to SFS for the adiabatic switchings, as we reviewed in Section~\ref{sec:adiabatic-smallgap}, and also for the near-adiabatic switchings considered in~\cite{Parry:2025wub}. 

We now give the proof of Theorem~\ref{prop:necessary_sign_changes}.  

\begin{proof}

To begin, we note by (iii) that $\widehat{\chi}_\lambda(0) = \Vert \chi_\lambda \Vert_1$, where $\Vert \cdot \Vert_1$ denotes the $L^1(\mathbb{R})$ norm.

We first show that $\widehat{\chi}_\lambda(0)$ grows proportionally to $\lambda^{1/2}$ as $\lambda \to \infty$. 
Setting $V(E, \omega)=1$ in~\eqref{eq:Vlimit}, it follows from Proposition \ref{prop:chilambda} 
that $\Vert \chi_\lambda \Vert^2 \geq A \lambda$ for sufficiently large $\lambda$, for a $\lambda$-independent positive constant~$A$. 
From (i) and~(iii), we have $\Vert \chi_\lambda \Vert^2 \leq M \lambda^{1/2} \Vert \chi_\lambda \Vert_1 = M \lambda^{1/2} \widehat{\chi}_\lambda(0)$. Combining, we find $\widehat{\chi}_\lambda(0) \geq B \lambda^{1/2}$, where $B:=A/M>0$, for sufficiently large~$\lambda$. Note that (iii) was used here in an essential way. 

We next show that $|\widehat{\chi}_\lambda(\omega)|$ grows proportionally to $\lambda^{1/2}$ not just at $\omega=0$ but on an interval of width of order $1/\lambda$ around $\omega=0$. 
From the absolute integrability of $\tau\chi_\lambda(\tau)$ in~(ii), it follows that $\widehat{\chi}_\lambda'$ is continuous. 
By the mean value theorem, we therefore have
\begin{equation}
\left| \widehat{\chi}_\lambda(\omega) - \widehat{\chi}_\lambda(0) \right| \leq |\omega| \sup_{s\in I_\omega}\left|\widehat{\chi}_\lambda'(s) \right| \,,
\end{equation}
where $I_\omega$ stands for $[0,\omega]$ for $\omega>0$ and $[\omega,0]$ for $\omega<0$. 
Now, from (ii) and (iii) it follows that 
\begin{equation}
\left|\widehat{\chi}_\lambda'(\omega) \right| \leq \int_{-\infty}^\infty \mathrm{d} \tau \, |\tau| |\chi_\lambda(\tau)| \leq C\lambda \widehat{\chi}_\lambda(0),
\end{equation}
using $\Vert \chi_\lambda \Vert_1 = \widehat{\chi}_\lambda(0)$. 
Therefore, on the interval $|\omega| \leq 1/(2C\lambda)$, we find $\left| \widehat{\chi}_\lambda(\omega) - \widehat{\chi}_\lambda(0) \right| \leq \frac{1}{2}\widehat{\chi}_\lambda(0)$ and consequently $\left| \widehat{\chi}_\lambda(\omega) \right|\ge \tfrac{1}{2}\widehat{\chi}_\lambda(0)$.
On this interval, we hence have
\begin{align}
\left| \widehat{\chi}_\lambda(\omega) \right| 
\geq \frac{1}{2}\widehat{\chi}_\lambda(0)
\geq \frac{B}{2} \lambda^{1/2} \,,
\end{align}
using the previously established inequality
$\left| \widehat{\chi}_\lambda(0) \right|\geq B \lambda^{1/2}$.

We now combine these observations to show that the SFS condition \eqref{eq:SFS} is not satisfied. 
Recall that $\lambda(E) \to \infty$ and $|E|\lambda(E) \to 0$ as $E\to 0$, by assumption. 
Let $|E|$ be so small that (i), (ii) and (iii) are satisfied with $\lambda = \lambda(E)$. Let further $|E|$ be so small that $|E|\lambda(E)\leq 1/(2C)$, and hence $|E|\leq 1/\bigl(2C\lambda(E)\bigr)$. We then find
\begin{align} \label{eq:SFS_not_satisfied}
\frac{1}{|E|\lambda(E)}\int_{-|E|}^{|E|}\mathrm{d} \omega|\widehat{\chi}_{\lambda(E)}(\omega)|^2 
\geq  \frac{1}{|E|\lambda(E)} \times 2|E| \times \frac{B^2}{4}\lambda(E) 
= \frac{B^2}{2} 
\,, 
\end{align}
using the length $2|E|$ of the integration interval and the bound $\left| \widehat{\chi}_\lambda(\omega) \right|\geq \frac{B}{2}\lambda^{1/2}$ over the integration interval. 
As the rightmost expression in \eqref{eq:SFS_not_satisfied} is a positive constant, \eqref{eq:SFS} is not satisfied. 
\end{proof}

\section{Conclusions} \label{sec:conclusions}

We have reviewed a curiosity of the circular motion Unruh effect for a massless scalar field in $2+1$ dimensions, namely that the effective temperature experienced by a UDW detector vanishes in the limit of long interaction time and small energy gap, when the limits are taken successively in either order~\cite{Parry:2025wub}. This limits the utility of the effective temperature as a quantifier of the circular motion Unruh effect for small detector gap, and the limitation may be especially relevant for recent analogue spacetime proposals to observe the circular motion Unruh effect in 
Bose-Einstein condensates \cite{Gooding:2020scc,Gooding:2025tfp} and 
superfluid Helium thin-films~\cite{Bunney:2023ude}. 

To address this curiosity, both as a theoretical observation and as a practical guide for the design of analogue spacetime experiments, we introduced in \cite{Parry:2025wub} asymptotically scaled switching families (ASSFs) for the UDW detector, and we identified the small frequency suppression (SFS) condition that is necessary and sufficient to obtain a nonzero effective temperature in a simultaneous limit of increasing interaction time and diminishing energy gap. 
In \cite{Parry:2025wub}, we also gave examples of families satisfying the SFS condition, all of which had the feature that the switching function changes sign. Importantly, once the ASSF satisfies SFS, the recovered temperature is independent of the detailed form of the switching family, depending only on the parameters of the circular motion in a manner consistent with the linear acceleration Unruh effect. In this sense, this temperature is a property inherent to the system,
revealed by our construction of suitably scaled switchings. 

As further discussed in~\cite{Parry:2025wub}, switching functions that change sign may be realised in practice, for example in relativistic entanglement harvesting protocols involving electromagnetic fields \cite{Lindel:2023rfi} or Bose-Einstein condensates~\cite{Gooding:2023xxl}. In these protocols, two causally disconnected local laser pulses can become entangled through their separate interactions with the quantum field of interest. The laser pulses play the role of Unruh-DeWitt type detectors, and the switching is determined by the coherent amplitude of each pulse, which need not have a fixed sign. 

In the present paper, as well as reviewing the results of~\cite{Parry:2025wub}, we have shown that a sign change in the switching function is a \emph{necessary\/} condition to obtain a nonzero limiting effective temperature, under certain technical boundedness and localisation conditions. It seems highly likely to us that the necessity of the sign change will hold even under weaker conditions, and we are aware of no examples in which a nonzero limiting effective temperature would be obtained for an ASSF with a uniform sign. Our technical conditions have the merit of making the necessity argument relatively straightforward, and we leave potential refinement of these conditions subject to future work. 

We end with two comments on directions of further development. 

First, recall that our observational scenario was one in which the detector's final state is observed but the final state of the ambient field is not. A~complementary scenario in which a local field observable is observed, but the detector's final state is not, has been recently addressed in~\cite{Wilkinson:2025jrw,Wilkinson:phdthesis}, with a review of earlier work. Scenarios that ask for both a local field observable and the detector's final state to be observed are a generalisation of obvious interest, and will need to be addressed within a framework that incorporates relativistic causality of observations in pairs of distinct spacetime regions~\cite{Fewster:2018qbm,Polo-Gomez:2025nuy}. 

Second, recall that our detector was spatially pointlike. This idealisation has certain limitations -- in particular, perturbation theory can be expected to encounter singularities in higher orders, with details depending on the observational scenario and on the spacetime dimension. The adequacy of the point-like detector idealisation would hence require further investigation in order to apply our results to experiment. One way to ameliorate the singularities is to introduce a spatially smeared profile for the detector, still modelled as a quantum mechanical system, but this has the cost of acausal propagation across the detector~\cite{deRamon:2021nry}; the acausality is however absent if the smeared detector is modelled within quantum field theory as sketched in~\cite{Fewster:2018qbm}. We leave such discussions to future work.

\section*{Acknowledgements}
JL thanks the organisers of the Jerzy Lewandowski Memorial Conference, Warsaw, 15–19 September 2025, for the invitation and the hospitality. 
The work of CJF was supported by the Engineering and Physical Sciences Research Council [grant number EP/Y000099/1]. 
The work of JL was supported by United Kingdom Research and Innovation Science and Technology Facilities Council [grant numbers ST/S002227/1, ST/T006900/1, ST/Y004523/1 and ST/C004235/1].
The authors have benefited from the activities of COST Action CA23115:
Relativistic Quantum Information, funded  by COST (European Cooperation in Science and Technology). 
For the purpose of open access, the authors have applied a CC BY public copyright licence to any Author Accepted Manuscript version arising.

\printbibliography

\end{document}